\newif\ifreview
\newif\ifarxiv

\reviewfalse
\arxivtrue

\ifreview
    \documentclass[journal,12pt,onecolumn,draftclsnofoot]{IEEEtran}
\else
    \documentclass[journal,twoside,web]{ieeecolor}
\fi 

\usepackage{generic}
\usepackage{cite}
\usepackage{amsmath,amssymb,amsfonts}
\usepackage{algorithmic}
\usepackage{graphicx}
\usepackage{algorithm,algorithmic}
\usepackage{hyperref}
\hypersetup{hidelinks}
\usepackage{textcomp}
\usepackage{siunitx}
\usepackage{tabularx}

\usepackage[export]{adjustbox}

\usepackage{balance}

\usepackage{xcolor}
\definecolor{BLUE}{rgb}{0,0,1}

\usepackage[utf8]{inputenc}
\usepackage{pgfplots}
\DeclareUnicodeCharacter{2212}{−}
\usepgfplotslibrary{groupplots,dateplot}
\usetikzlibrary{patterns,shapes.arrows}
\pgfplotsset{compat=newest}
\pgfplotsset{every tick label/.append style={font=\scriptsize},
        every y tick scale label/.append style={yshift=-1em}}

\renewcommand{\vec}[1]{#1}
\newcommand{\xdot}{\dot{x}}
\newcommand{\x}{x}
\renewcommand{\u}{u}
\newcommand{\f}{f}
\newcommand{\g}{g}

\renewcommand{\u}{\vec{u}}

\newcommand{\set}[1]{\mathbb{#1}}
\newcommand{\R}{\mathbb{R}}
\newcommand{\N}{\mathbb{N}}
\newcommand{\K}{\mathcal{K}}

\newcommand{\xinactive}{\set{D}_{\rho \neq 1}}
\newcommand{\xineps}{\set{D}_{\epsilon}}

\usepackage{stackengine}

\newtheorem{theorem}{Theorem}
\newtheorem{corollary}{Corollary}

\newtheorem{proposition}{Proposition}
\newtheorem{definition}{Definition}
\newtheorem{lemma}{Lemma}

\def\BibTeX{{\rm B\kern-.05em{\sc i\kern-.025em b}\kern-.08em
    T\kern-.1667em\lower.7ex\hbox{E}\kern-.125emX}}
\ifreview
\else
\fi

\begin{document}
\title{Preventing Inactive CBF Safety Filters Caused by Invalid Relative Degree Assumptions}

\author{Lukas Brunke, \IEEEmembership{Graduate Student Member, IEEE}, Siqi Zhou, \IEEEmembership{Member, IEEE},\\ and Angela P. Schoellig, \IEEEmembership{Member, IEEE}
\thanks{Lukas Brunke and Angela P. Schoellig are with the Learning Systems and Robotics Lab, 
Technical University of Munich, 80333 Munich, Germany, also with the University of Toronto Institute for Aerospace Studies, North York, ON M3H 5T6, Canada, 
and also with the Vector Institute for Artificial Intelligence, Toronto, ON M5G 0C6, Canada~(e-mail: lukas.brunke@tum.de; angela.schoellig@tum.de).
}
\thanks{Siqi Zhou is with the Learning Systems and Robotics Lab, 
Technical University of Munich, 80333 Munich, Germany~(e-mail: siqi.zhou@tum.de).}
\thanks{This work has been supported by the Robotics Institute Germany, funded by BMBF grant 16ME0997K, and the European Union’s Horizon 2024 research and innovation programme under the Marie Skłodowska-Curie grant agreement No. 101155035 (SSDM).}
}

\maketitle

\begin{abstract}
Control barrier function~(CBF) safety filters emerged as a popular framework to certify and modify potentially unsafe control inputs, for example, provided by a reinforcement learning agent or a non-expert user. Typical CBF safety filter designs assume that the system has a uniform relative degree. This assumption is restrictive and is frequently overlooked in practice. When violated, the assumption can cause the safety filter to become inactive, allowing large and possibly unsafe control inputs to be applied to the system. In discrete-time implementations, the inactivity issue is often manifested as chattering close to the safety boundary and/or constraint violations. In this work, we provide an in-depth discussion on the safety filter inactivity issue, propose a mitigation strategy based on multiple CBFs, and derive an upper bound on the sampling time for safety under sampled-data control. The effectiveness of our proposed method is validated through both simulation and quadrotor experiments.
\end{abstract}


\section{Introduction}
\label{sec:introduction}

\IEEEPARstart{S}{afety} is paramount in any real-world application of control systems, such as robotics~\cite{DSL2021}.
A safe set encodes states the system should stay within~(i.e., a quadrotor avoiding collision with the floor, see ~\autoref{fig:money-figure}).
While most control systems are designed with a focus on safety, many recent learning-based control methods~(e.g., reinforcement learning) typically do not provide safety guarantees~\cite{DSL2021}. Similarly, teleoperating a control system may lead to unsafe scenarios, especially when the user is not an expert. 
In such cases, the system's safety can be ensured by employing a safety filter, which certifies and, if necessary, modifies an unsafe control input before it is applied to the control system. 

A popular line of work on safety filters uses control barrier functions~(CBFs)~\cite{jaime-survey-2024}. 
Based on Nagumo's theorem~\cite{Nagumo1942berDL}, CBFs~\cite{WIELAND2007462, ames2019a} provide a scalar condition to check control invariance, which can be efficiently incorporated into a safety filter.  
CBF safety filters gained popularity due to (\textit{i}) their applicability to control affine systems, resulting in an efficiently-solvable quadratic program~(QP), and (\textit{ii}) the removal of overly conservative restrictions that kept the system inside shrinking safe sets~\cite{ames2019a}.    
Furthermore, CBFs have already been successfully applied to many real-world safety-critical systems~(e.g., safe quadrotor flight~\cite{brunke-lcss-2024},
or robot swarms~\cite{Glotfelter2017}). 

\begin{figure}[t]
    \definecolor{mypurple}{RGB}{190,68,189}
    \definecolor{lightgray}{RGB}{200,200,200}
    \centering
    \begin{tikzpicture}
    \ifreview
        \node[anchor=south west,inner sep=0] (image) at (0,0) {\includegraphics[width=21pc, trim={130 35 190 260}, clip]{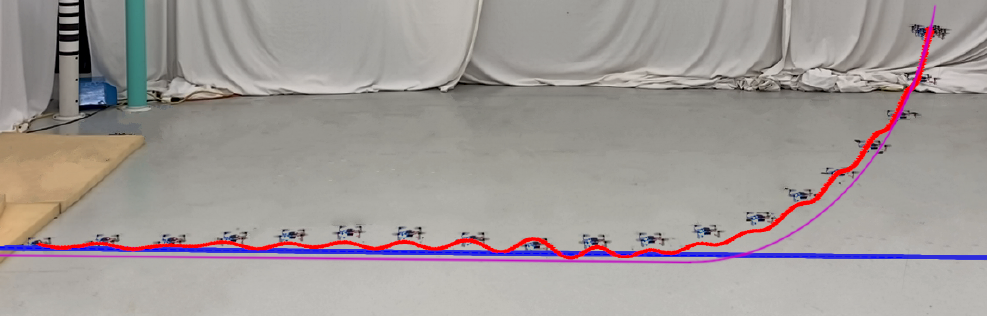}};
        \draw[white, fill, fill opacity=1.0, draw opacity =0.0] (0.0,2.1) rectangle (6.3, 2.9);
    \else 
        \node[anchor=south west,inner sep=0] (image) at (0,0) {\includegraphics[width=\columnwidth, trim={130 35 190 260}, clip]{figures/tac-money-figure-removed-noise-single_0905.pdf}};
        \draw[white, fill, fill opacity=1.0, draw opacity =0.0] (0.0,2.1) rectangle (5.4, 2.9);
    \fi
    \node [anchor=west] (note) at (0.05,2.6) {Standard Single-CBF Safety Filter};
    \node [anchor=west] (note) at (0.05,2.3) {\scriptsize Chattering and constraint violations may occur};
    \node [fill= lightgray, fill opacity=1.0, anchor=west] (note) at (6.6,0.25) {\textcolor{blue}{\scriptsize Safe set boundary}};
    \node [fill= lightgray,anchor=west] (note) at (4.8,1.65) {\textcolor{red}{\scriptsize Closed-loop trajectory}};
    \node [fill= lightgray,anchor=west] (note) at (6.4,2.6) {\textcolor{red}{\scriptsize Starting point}};
    \node [anchor=west] (start) at (8.1, 2.55) {\scriptsize\textcolor{red}{$\times$}}; 
    \node [fill= lightgray,anchor=west] (note) at (7.5,0.8) {\textcolor{mypurple}{\scriptsize Reference}};
\end{tikzpicture}\\%
\begin{tikzpicture}
    \ifreview
        \node[anchor=south west,inner sep=0] (image) at (0,0) {\includegraphics[width=21pc, trim={130 35 190 260}, clip]{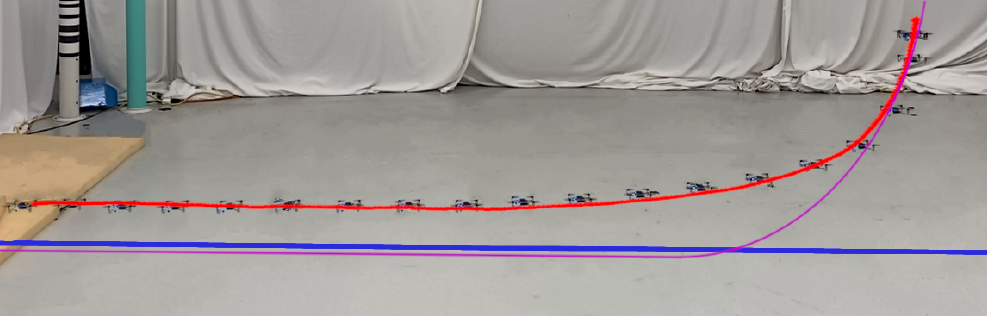}};
        \draw[white, fill, fill opacity=1.0, draw opacity =0.0] (0.0,2.1) rectangle (6.3, 2.9);
    \else 
        \node[anchor=south west,inner sep=0] (image) at (0,0) {\includegraphics[width=\columnwidth, trim={130 35 190 260}, clip]{figures/tac-money-figure-removed-noise-multi_0905.pdf}};
        \draw[white, fill, fill opacity=1.0, draw opacity =0.0] (0.0,2.1) rectangle (5.4, 2.9);
    \fi
    \node [anchor=west] (note) at (0.05,2.6) {Proposed Multi-CBF Safety Filter};
    \node [anchor=west] (note) at (0.05,2.3) {\scriptsize Safe operation is guaranteed};
    \node [fill= lightgray, fill opacity=1.0, anchor=west] (note) at (6.6,0.25) {\textcolor{blue}{\scriptsize Safe set boundary}};
    \node [fill= lightgray,anchor=west] (note) at (4.8,1.65) {\textcolor{red}{\scriptsize Closed-loop trajectory}};
    \node [fill= lightgray,anchor=west] (note) at (6.4,2.6) {\textcolor{red}{\scriptsize Starting point}};
    \node [anchor=west] (start) at (8.00, 2.55) {\scriptsize\textcolor{red}{$\times$}}; 
    \node [fill= lightgray,anchor=west] (note) at (7.5,0.8) {\textcolor{mypurple}{\scriptsize Reference}};
\end{tikzpicture}%
    \caption{
    Real-world quadrotor demonstration with nonlinear dynamics~(safe set boundary in blue, trajectory in red). Single CBF (top) results in an inactive safety filter, causing chattering and safe set violations. Our multi-CBF approach (bottom) successfully keeps the safety filter active and the system safe. See Figure~\ref{fig:multi-cbf-real-input} for details.
    }
    \label{fig:money-figure}
\end{figure}

The relative degree of a system plays an important role in CBF safety filters. 
The relative degree of a system indicates 
which time derivative of the system's output can be directly impacted by the control input~\cite{khalil2002}.  
This local property may not be constant over a nonlinear system's domain.
An active first-order CBF safety filter, where the CBF condition further constrains the set of feasible control inputs, 
requires the system's relative degree to be one in the entire domain~\cite{ames2017}.
If this requirement is not met, the control input cannot impact the first time derivative of the CBF at this state, and the CBF safety filter reduces to projecting control inputs onto the static feasible input set. 

This relative degree assumption has been relaxed by introducing higher-order CBFs, applicable to systems with a relative degree greater than one~\cite{nguyen-exp-cbf-2016, xiao2021high}.
These methods, however, still assume a well-defined, uniform relative degree over the entire domain. 
The assumption of uniform relative degree for first- and higher-order CBFs is not always verified in practice
and poses limitations. While the author in~\cite{JANKOVIC2018359} relaxes this assumption by relying on drift dynamics, the authors in~\cite{Alyaseen2025TAC} show that control policy boundedness and continuity can degrade where this assumption fails, and prove it is necessarily violated for compact safe sets defined by smooth CBF superlevel sets. In~\cite{Tan2022TAC}, the authors propose requiring uniform relative degree only on the safe set boundary, which is still limiting as it does not apply to simple cases like linear dynamics with quadratic CBFs~\cite{Alyaseen2025TAC}. Recently, the work by~\cite{lukas-acc-2024} proposed practical approaches for CBFs with non-uniform relative degree, though without theoretical analysis.
Building on the insights from these works, we formally define conditions to prevent inactive CBF safety filters and develop our mitigation strategy. 

Another key challenge in applying CBF safety filters is their discrete-time implementation. Though designed for continuous time, real-world systems require applying control inputs at discrete intervals due to delays and computation time, forming a sampled-data system~\cite{FRIDMAN20041441}. This can worsen issues like violating relative degree assumptions, as feasible large inputs in magnitude may persist during each sampling interval.
Two main approaches address this:
    First, discrete-time CBFs~\cite{Agrawal-RSS-17}, which ensure safety directly in discrete time but often involve nonconvex optimization, increasing computational load.
    Second, continuous-time CBFs with sampled-data adjustments, which either add conditions for fixed sampling~\cite{singletary2020-sampled-data, breeden2022-sampled} or adapt sampling online via self-triggering~\cite{yang2019, bahati2024}. These methods rely on bounding the deviation between the true trajectory and the last sampled state.
In this work, we use a fixed sampling time with a continuous-time formulation. 
This allows us to exploit the efficient QP formulation and provide safety guarantees for all time rather than just at discrete time steps.
Furthermore, we select the minimal feasible sampling time to track the uncertified control input as closely as possible.  

In this work, we consider the interplay of these challenges: how the violation of uniform relative degree assumptions leads to inactive safety filters and how their sampled-data implementation causes safe set violations and chattering~(i.e., high-frequency changes in the control input due to the safety filter)~\cite{
fpb-cdc2023}.
To prevent safe set violations, we present a theoretically sound mitigation method 
based on multiple CBFs in the safety filter.
While multiple CBFs have been considered in the literature, they have generally been used for different purposes, such as inter-agent collision avoidance in multi-agent systems~\cite{Glotfelter2017} or avoiding multiple obstacles~\cite{parwana2023feasible, aali2022}. Our work uniquely applies them to address inactive CBF safety filters caused by non-uniform relative degree in sampled-data systems. 
We summarize our contributions as follows:
\begin{enumerate}
    \item 
    We analyze CBF safety filter inactivity and examine how feasible input direction and magnitude at such states can cause safety violations or chattering in discrete time.
    
    \item 
    We propose a multi-CBF safety filter along with a synthesis procedure to avoid inactivity and thereby mitigate the resulting chattering and constraint violation issues.
    \item %
    We derive a theoretical upper bound on the sampling time for guaranteeing safety under sampled-data implementation.
    \item
    We validate our approach through simulation and real-world quadrotor experiments.
    
\end{enumerate}

The article is organized as follows: \autoref{sec:notation} introduces notation, \autoref{sec:problem} formulates the problem, and \autoref{sec:background} provides background. Our method is presented in \autoref{sec:method}, with evaluations in \autoref{sec:eval}, and conclusion in \autoref{sec:conclusion}.

\section{Mathematical Preliminaries}
\label{sec:notation}
In this section, we introduce the notation and general definitions used in the remainder of this work. 
Let $\set{Z}$, $\N$, and $\N_0$ denote the sets of integers, positive integers, and nonnegative integers, respectively. 
A set of consecutive integers is $\set{Z}_{a, b}$, where $a, b \in \set{Z}$ and $a \leq b$. The Euclidean norm is $\lVert \cdot \rVert$. 
We denote the diagonal matrix \( D \in \mathbb{R}^{n \times n} \) with diagonal entries given by a vector \( d \in \mathbb{R}^n \) as \( D = \operatorname{diag}(d) = \operatorname{diag}(d_1, \dots, d_n)\) and the diagonal entries as $d = \operatorname{diag}(D)$.
The set $\set{B}(x_0, r)$ 
denotes the closed $2$-norm ball of radius $r$ centered at $x_0 \in \R^n$.
The Lie derivative of a differentiable function $h$ along a vector field $f$ is denoted by $\mathcal{L}_f h(x) = \frac{\partial h(x)}{\partial x} f(x)$.

In this work, we consider a control-affine system:
\begin{equation}
\label{eq:nonlinear_affine_control}
	\xdot(t) = \f(\x(t)) + \g (\x(t)) \:\u(t)\, ,
\end{equation}
where $t \in \R$ is the time variable, $\x\in \set{D} \subset \R^n$ is the state of the system with $\set{D}$ being the set of admissible states, $\u\in \set{U} \subset \R^m$ is the input of the system with $\set{U}$ being the set of admissible inputs, and  $\f:\R^n\mapsto \R^n$ and $\g:\R^n\mapsto \R^{n\times m}$ are locally Lipschitz continuous functions. 
Hereafter, we omit the dependency on $t$ unless unclear from context.

\begin{definition}[Positively control invariant set] Let $\mathfrak{U}$ be the set of all bounded control signals
$\nu : \R_{\geq 0} \to \set{U}\,$. A set ${\set{C}\subseteq\set{D}}$ is a positively control invariant set for the control system in~\eqref{eq:nonlinear_affine_control} if~$\: \forall \: {\x_0 \in\set{C}} \,,\: \exists \: \nu \in \mathfrak{U} \,, \: \forall \: t \in \set{T}_{\x_0}^+ \,,\: \phi(t, \x_0, \nu) \in \set{C}$, where $\phi(t, \x_0, \nu)$ is the system's phase flow starting at $\x_0$ under the control signal $\nu$, and $\set{T}_{\x_0}^+$ is the maximum time interval.
\end{definition}

In this work, we also refer to a positively control invariant set as a safe set and to positive control invariance as safety. 

\begin{definition}[Extended class-$\K$ function~\cite{ames2019a}]
\label{def:kappa_inf_extended}
    A function $\gamma: \R \to \R$ is said to be of class-$\K_e$
    if it is continuous, $\gamma(0) = 0$, and strictly increasing. 
\end{definition}
\begin{definition}[Least relative degree and relative degree]
\label{def:rel-degree}
 Consider a set $\set{D} \subset \R^n$ and a system consisting of the dynamics equation in \eqref{eq:nonlinear_affine_control} and the output equation $y = h(x)$, where $h$ is $\rho^{\text{th}}$-order differentiable. The system is said to have a \textit{least relative degree} of $\rho \in \set{Z}_{1, n}$ over $\set{D}$ if $\mathcal{L}_g \mathcal{L}_f^{i} h(x) = 0$ for $i \in \set{Z}_{0, \rho -2}$~\cite[Def. 4]{Tan2022TAC}.
    Furthermore, the system is said to have a \textit{relative degree} of $\rho \in \set{Z}_{1, n}$ over $\set{D}$ if  $\mathcal{L}_g \mathcal{L}_f^{\rho - 1} h(x) \neq 0$ for all $x \in \set{D}$ is additionally satisfied~\cite[Def. 13.2]{khalil2002}.
\end{definition}

\section{Problem Formulation}
\label{sec:problem}
We consider the control-affine system in~\eqref{eq:nonlinear_affine_control}. We assume that the domain $\set{D}$ is an open set, and $f$ and $g$ have Lipschitz constants $L_f$ and $L_g$, respectively.
We are given a bounded desired safe set $\set{X} \subseteq \set{D}$, in which we want to keep the system for all future time $t \geq 0$.
We aim to find a compact inner approximation of the desired safe set, $\set{C}_{\set{Z}_{1, K}} \subseteq \set{X}$, where $K \in \N$, $\set{C}_{\set{Z}_{1, K}} = \{ x \in \R^n \,\vert\, h_i(x) \geq 0\,, \forall i \in \set{Z}_{1, K} \}$ that is control invariant under non-empty compact polytopic input constraints $\set{U} = \{ u \in \R^m \vert A_u u \leq b_u \}$ with continuously differentiable functions  $h_i$ for all $i \in \set{Z}_{1,K}$. 
The boundary of the safe set is $\partial \set{C}_{\set{Z}_{1,K}} =  \{ x \in \R^n \,\vert\, \min_{i \in \set{Z}_{1, K}} h_i(x) = 0 \}$. 

Our goal is to augment a given, potentially unsafe state-feedback controller~$\pi: \R^n \to \R^m$ with a safety filter $\pi_{\text{cert}}: \R^n\times \R^m \to \R^m$ such that the system is safe~(i.e., the system's state~$\x$ stays inside a safe set $\set{C}_{\set{Z}_{1, K}}$ if it starts inside of $\set{C}_{\set{Z}_{1, K}}$).
We assume that the controller $\pi$ is locally Lipschitz continuous and has a Lipschitz constant $L_{\pi}$.

In this work, we aim to guarantee safety for a sampled-data system, where the control input can only be updated discretely at fixed time intervals $t_k = k \Delta t$ with $k \in \N_0$ and $\Delta t > 0$. In a zero-order-hold (ZOH) fashion, the control input is
\begin{equation}
\label{eq:sampled-input}
    u(t_k + \tau) = u_{\text{cert},k} = \pi_{\text{cert}} (x(t_k), u_k), \forall \tau \in \left[0, \Delta t\right)
\end{equation}
with the uncertified control input $u_k = \pi(x(t_k))$ and $u_{\text{cert},k}$ is the certified control input. 
This setting is motivated by real-world systems, where 
control inputs cannot be applied arbitrarily fast to the system.

A safety filter based on a single CBF~($K = 1$) can contain states on the boundary of $\set{C}_{\{ 1\}}$ for which the safety filter will become inactive. 
This is a known issue that reduces the CBF safety filter to projecting the uncertified control input $\pi(x)$ onto the set $\set{U}$, and the CBF has no effect. 
This inactivity can often lead to safety constraint violations. 
Therefore, the goal in our work is to propose a mitigation strategy that resolves the inactivity issue. 
\section{Background}
\label{sec:background}
This section introduces the relevant definitions and background on CBF-based safety filters to facilitate our discussion.

\begin{definition}[Control barrier function~(CBF)~\cite{Tan2022TAC}]
\label{def:cbf-tan}
    Consider the control system in~\eqref{eq:nonlinear_affine_control} and a first-order differentiable function $h : \R^n \to \R$. The function $h$ is called a CBF, if there exists a class-$\K_e$ function $\gamma$, an open set $\set{D}$ with $\set{C} \subset \set{D} \subset \R^n$, where $\set{C}$ is the zero superlevel set of $h$, the system is of least relative degree $\rho = 1$, and for all $x \in \set{D}$
    \begin{equation}    \label{eq:cbf_lie_derivative}
		\max_{\u \in \set{U}} \left[\mathcal{L}_\vec{f} h(\x) + \mathcal{L}_\vec{g} h(\x) \u \right] \geq - \gamma(h(x)) \,.
	\end{equation}
\end{definition}

Using a CBF, we can define an input set
\begin{equation}
\label{eq:cbf-control-input-set}
 \set{U}_{h,\gamma}(x) = \{u\in \set{U} \:\vert \:  \mathcal{L}_\vec{f} h(\x) + \mathcal{L}_\vec{g} h(\x) \u \geq - \gamma(h(x))\} \,
\end{equation}
that renders the system safe~\cite{ames2019a}.
\begin{theorem}[Control invariance of $\set{C}$~\cite{Tan2022TAC}]
    Consider a CBF $h$ as defined in~\autoref{def:cbf-tan}. Then, any locally Lipschitz continuous controller $\pi : \R^n \to \R^m$ such that $\pi(x) \in \set{U}_{h,\gamma} (x)$ for all $x \in \set{C}$ will render the set $\set{C}$ positively control invariant for the system~\eqref{eq:nonlinear_affine_control}. 
\end{theorem}

For a 
policy $\pi(x)$ that is not initially designed to be safe, one can formulate a QP to modify the control input such that the system is guaranteed to be safe~\cite{ Ames2014}:
	\begin{align}
    \label{eqn:cbf_qp}
	\pi_{\text{cert}}(x, u_k) = &\underset{\u \in \set{U}_{h, \gamma}(x)}{\text{argmin}}  \quad \frac{1}{2} \lVert \u - u_k \rVert^2 
	\end{align}
Intuitively, a safety filter finds an input in $\set{U}_{h, \gamma}(x)$ that best matches $u_k$, 
based on a chosen distance measure~(e.g., the Euclidean norm in~\eqref{eqn:cbf_qp}).

We model the discretization error introduced by the ZOH as a disturbance. %
The closed-loop sampled data system is
\begin{equation}
\label{eq:sampled-data-sys}
    \dot{x}(t) = f(x(t)) + g(x(t)) \pi(x(t_k)) \,,
\end{equation}
where $t_k = k \Delta t$ with $k \in \N_0$. 
Based on the results in~\cite{Kolathaya2019}, a control policy $\pi(x)$ yields an input-to-state safe closed-loop system on $\set{C}_d = \left\{ x \in \R^n \, \vert \, h_d(x) = h(x) - d \geq 0\right\} \neq \emptyset$ for the sampled-data implementation in~\eqref{eq:sampled-data-sys} if the following holds:
\begin{equation}
    \label{eq:sampled-data-lie-derivative-cond}
    \begin{aligned}
        & \mathcal{L}_f h_{d}(x(t)) + \mathcal{L}_g h_{d}(x(t)) \pi(x(t_k)) \geq \\ 
        & \quad - \gamma(h_{d}(x(t))) - L_\pi M e( t, x(t_k), \pi(x(t_k)) \,,
    \end{aligned}
\end{equation}
where $t = t_k + \Delta t$, $M>0$ is such that $\lVert \mathcal{L}_g h(x) \rVert = \lVert \mathcal{L}_g h_d(x) \rVert \leq M$, and $e( t, x(t_k), \pi(x(t_k)) = \lVert x(t) - x(t_k) \rVert$. 
\begin{theorem}[Invariance of sampled-data systems~\cite{bahati2024}]\label{thm:inv-sampled-data}Let $h$ be a continuously differentiable function 
and the closed-loop sampled data system in~\eqref{eq:sampled-data-sys} be Lipschitz continuous with a state-feedback control policy $\pi$. Furthermore, let $\pi(x)$ yield an input-to-state safe closed-loop sampled data system on $\set{C}_d$. For each $d \in \left( 0, h_{\max} \right)$, where $h_{\max} = \max_{x \in \set{C}} h(x)$, there exists an upper bound $e_{\max} = -\frac{\gamma(-d)}{L_{\pi} M} > 0$ on the sample-and-hold error such that if $\lVert e \rVert_\infty \leq e_{\max}$, the set $\set{C}$ 
is positively control invariant under the control policy $\pi(x)$. 
\end{theorem}

In~\cite{bahati2024}, Lipschitz continuity of the certified control policy in the sampled-data implementation is assumed, which can be restrictive in practice. This assumption can be relaxed by tightening the CBF condition in~\eqref{eq:cbf_lie_derivative} based on the sampling time and the system dynamics to still achieve control invariance~\cite{breeden2022-sampled}.
We further note that while we focus on the sampled-data effect, the notion of input-to-state safety can be extended to other types of practical disturbances~(e.g., model mismatch and measurement noise)~\cite{Kolathaya2019}.
\section{Methodology}
\label{sec:method}
In this section, we begin by investigating the cause and consequence of inactive CBF safety filters. Based on the insights, we then propose a method to mitigate inactive safety filter issues 
and detail how safety can be achieved in practice.

\subsection{Inactive CBF Safety Filters}

\subsubsection{Conditions for Inactivity}
We refer to a CBF safety filter being inactive at state $x$ when the constraint~\eqref{eq:cbf_lie_derivative} has no effect on the resulting certified control policy $\pi_{\text{cert}}(x)$, i.e., the constraint is redundant with $\set{U}_{h,\gamma}(x) = \set{U}$. At such states, the safety filter reduces to a projection of the policy $\pi(x)$ onto the control input set $\set{U}$. 
Clearly, the safety filter is inactive at $x$ when $\mathcal{L}_g h(x) = 0$. Then one of the following is satisfied: \textit{(i)} $g(x) = 0$~(the system is autonomous at $x$),~\textit{(ii)} $\frac{\partial h(x)}{\partial x} = 0$, or \textit{(iii)} $\frac{\partial h(x)}{\partial x} \perp g(x)$~(the gradient of $h$ is orthogonal to every column in $g(x)$). 
The Lie derivative $\mathcal{L}_gh(x)$ being zero is related to the system's relative degree at $x$. 
For convenience, denote the states where the relative degree 
$\rho = 1$
as $\set{D}_{\rho = 1} = \{x \in \R^n \, \vert \, \mathcal{L}_g h(x) \neq 0 \}$. We also define the set of states where the relative degree is not equal to one as $\set{D}_{\rho \neq 1} = \set{D} \setminus \set{D}_{\rho = q}$.

Besides $\mathcal{L}_g h(x) = 0$, the safety filter may also be effectively inactive when $\mathcal{L}_g h(x)$ has small magnitude.
To show this, we rewrite the CBF condition as 
    $\alpha^\intercal(x) u \leq \beta(x) \,,$
where $\alpha^\intercal(x) = - \frac{\mathcal{L}_g h(x)}{\lVert \mathcal{L}_g h(x) \rVert}$ with $\lVert \alpha \rVert  = 1$ and $\beta(x) = \frac{\gamma(h(x)) + \mathcal{L}_f h(x)}{\lVert \mathcal{L}_g h(x) \rVert}$.
This reveals that every CBF yields a state-dependent affine control input constraint in standard affine form.
The support function $\sigma_{\set{U}}(\alpha) = \max_{u \in \set{U}} \alpha^\intercal u$ is a linear program that measures how far the set $\set{U}$ extends in the direction of $\alpha$. 
Consequently, the safety filter is inactive if $\beta(x) > \sigma_{\set{U}}(\alpha(x))$ or, equivalently, if $\lVert \mathcal{L}_g h(x) \rVert < \frac{\gamma(h(x)) + \mathcal{L}_f h(x)}{\sigma_{\set{U}}(\alpha(x))} =: \epsilon(x)$.  

\subsubsection{Impacts of Inactive CBF Safety Filters}

Inactive safety filters have several practical implications.
First, as highlighted above, for states $\bar{x} \in \set{C}$ such that $\mathcal{L}_g h(\bar{x}) = 0$ the gradient of $h$ is orthogonal to every column $g(\bar{x})$. 
The input matrix $g(\bar{x})$ projects the control input $u$ onto a subspace of $\set{D}$ spanned by the columns of $g(\bar{x})$ that is tangential to the level set of $h(x)$ at $\bar{x}$. 
For any convex zero superlevel set $\set{C}_d$, where $d = h(\bar{x})$, the input's  contribution to the system dynamics $g(\bar{x})u$ at $\bar{x}$ points outside of $\set{C}_d$ for any $u \neq 0$. In particular, for any $\bar{x} \in \partial \set{C}$ and $u\neq 0$, we have that $g(x) u$ will point outside of the safe set $\set{C}$, see~\autoref{fig:vector-epsilon}. 
The drift term $f(x)$ may counteract or support this direction, depending on the system dynamics. 
We note that if any element of $\mathcal{L}_g h(x)$ is 0, the same issue would also arise for the corresponding input channel.

\begin{figure}
    \centering
    \begin{adjustbox}{trim=0 1.9cm 0 0,clip}
    \begin{tikzpicture}[->, thick, scale=2]
    \definecolor{green01270}{RGB}{0,127,0}
    \def\radius{1.0}
    \def\scale{0.6}
    \def\deltat{0.75}
    \def\sqrttwo{1.41421356237}
    \def\sqrtthree{1.73205080757}
    \def\sinfifteen{0.2588190451}
    \def\cosfifteen{0.96592582629}
    \def\px{-\radius * 1}
    \def\py{\radius * 0}
    \def\hx{- \scale * 2 * \px}
    \def\hy{- \scale * 2 * \py}
    \def\tconexp{- \hy}
    \def\tconeyp{\hx}
    \def\tconexm{\hy}
    \def\tconeym{-\hx}
    \def\gx{0.0}
    \def\gy{1}
    \def\gxd{\deltat * \deltat * 0.5}
    \def\gyd{\deltat}
    \def\fx{\py}
    \def\fy{0}
    \def\fxd{\px + \py * \deltat}
    \def\fyd{\py}
    \def\up{1.2}
    \def\um{-2}
    \def\xdotx{\fx + \gx * \up}
    \def\xdoty{\fy + \gy * \up}
    \def\xnextx{\fxd + \gxd * \up}
    \def\xnexty{\fyd + \gyd * \up}
  \fill[gray!30, opacity=0.8] (0,\radius) arc[start angle=90, end angle=270, radius=\radius] -- cycle;

    \draw[-, dashed, draw=gray] (\px,\py) -- (\px+\tconexp,\py+\tconeyp) node[pos=1,right] {}; 
    \draw[-, dashed, draw=gray] (\px,\py) -- (\px+\tconexm,\py+\tconeym) node[pos=1,right] {}; 

  \draw[-, draw=green01270] (0,0) -- (-1.3 * \radius,0) node[below, midway, text=green01270] {$\mathcal{L}_g h(x) = 0$};
  \draw[->, draw=blue, dashed, opacity=0.5] (\px,\py) -- (\px+\hx,\py+\hy) node[above, text=blue, opacity=0.5] {$\nabla h(x)$};
  \draw[->, draw=red, opacity=0.5] (\px,\py) -- (\px+\gx * \up,\py+\gy * \up) node[left, text=red, , opacity=0.5] {$B\bar{u}$};
  \draw[->, draw=red] (\px,\py) -- (\px+\gxd * \up,\py+\gyd * \up) node[midway, right, text=red] {$\hat{B}\bar{u}= \Delta x$};



\end{tikzpicture}
  \end{adjustbox}
%
 \hspace{0.5cm} 
 \begin{adjustbox}{trim=0 1.4cm 0 0,clip}
    \begin{tikzpicture}[->, thick, scale=2]
    \definecolor{green01270}{RGB}{0,127,0}
    \def\radius{1.0}
    \def\scale{0.6}
    \def\deltat{0.75}
    \def\sqrttwo{1.41421356237}
    \def\sqrtthree{1.73205080757}
    \def\sinfifteen{0.2588190451}
    \def\cosfifteen{0.96592582629}
    \def\px{-\radius * \cosfifteen}
    \def\py{\radius * \sinfifteen}
    \def\hx{- \scale * 2 * \px}
    \def\hy{- \scale * 2 * \py}
    \def\tconexp{- \hy}
    \def\tconeyp{\hx}
    \def\tconexm{\hy}
    \def\tconeym{-\hx}
    \def\gx{0.0}
    \def\gy{1}
    \def\gxd{\deltat * \deltat * 0.5}
    \def\gyd{\deltat}
    \def\fx{\py}
    \def\fy{0}
    \def\fxd{\px + \py * \deltat}
    \def\fyd{\py}
    \pgfmathsetmacro{\hTfvalue}{\hx * \fx + \hy * \fy}
    \pgfmathsetmacro{\hTgvalue}{\hx * \gx + \hy * \gy}
    \pgfmathsetmacro{\up}{-\hTfvalue / \hTgvalue}
    \def\um{-2}
    \def\xdotx{\fx + \gx * \up}
    \def\xdoty{\fy + \gy * \up}
    \def\xnextx{\fxd + \gxd * \up}
    \def\xnexty{\fyd + \gyd * \up}
  \fill[gray!30, opacity=0.8] (0,\radius) arc[start angle=90, end angle=270, radius=\radius] -- cycle;

    \draw[-, dashed, draw=gray] (\px,\py) -- (\px+\tconexp,\py+\tconeyp) node[pos=1,right] {}; 
    \draw[-, dashed, draw=gray] (\px,\py) -- (\px+\tconexm,\py+\tconeym) node[pos=1,right] {}; 
    
  \draw[-, draw=green01270] (0,0) -- (-1.3 * \radius,0) node[below, midway, text=green01270] {$\mathcal{L}_g h(x) = 0$};
  \draw[->, draw=blue, opacity=0.5, dashed] (\px,\py) -- (\px+\hx,\py+\hy) node[above, text=blue, opacity=0.5] {$\nabla h(x)$};
  \draw[->, draw=red, opacity=0.5] (\px,\py) -- (\px+\gx * \up,\py+\gy * \up) node[left, text=red, opacity=0.5] {$Bu$};
  \draw[->, draw=red] (\px,\py) -- (\px+\gxd * \up,\py+\gyd * \up) node[midway, left, text=red] {$\hat{B}u$};

  \draw[->, thick, draw=magenta] (\px,\py) -- (\fxd,\fyd) node[right, text=magenta] {$\hat{A}x$};

  \draw[-, dashed] (\fxd,\fyd) -- (\xnextx,\xnexty) node[midway,right] {};
  \draw[-, dashed] (\px+\gxd * \up,\py+\gyd * \up) -- (\xnextx,\xnexty) node[midway,right] {};
  \draw[->, thick] (\px,\py) -- (\xnextx,\xnexty) node[pos=1,right] {$\Delta{x}$}; 

\end{tikzpicture}
    \end{adjustbox}
    \caption{
The effect of $\lVert \mathcal{L}_g h(x) \rVert$ on safety filter activity for a double integrator system on an ellipsoidal safe set (gray). We use $\Delta t = \SI{0.75}{\second}$ and uncertified policy $\pi(x) = \bar{u} = \max_{\set{U}} \lVert u \rVert$. When $\mathcal{L}_g h(x) = 0$~(safety filter inactive, left), $\pi(x) = \bar{u}$ is allowed. When $\lVert \mathcal{L}_g h(x) \rVert \neq 0$~(right), $\bar{u}$ is not allowed, though the certified input $u$ remains large. Both cases can cause safe set violations in sampled-data implementation.
    }
    \label{fig:vector-epsilon}
\end{figure}

Second, if the safety filter is inactive at $\bar{x}$, any control policy $\pi(\bar{x}) = u_{\text{cert}} \in \set{U}$ is feasible, see~\autoref{fig:vector-epsilon}. 
Again, if any element of $\mathcal{L}_g h(x)$ is 0, then the associated element of the input vector is also not constrained by the CBF. 
An example of a feasible but aggressive control policy is a policy that tries to maximize the magnitude of $g(x)u$ with $\pi_{+}(x) = \arg \max_{u \in \set{U}} \lVert g(x) u \rVert \in \set{U}$.
The control policy $\pi_{+}(\bar{x})$ is certified by the optimization problem in~\eqref{eqn:cbf_qp} since it is feasible.
 If the control input set $\set{U}$ is large, i.e., the control input constraints are loose, the control input $u_{\text{cert}} = \pi_{+}(\bar{x})$ is also large in magnitude.
Since the set $\set{U}$ is time- and state-invariant, it already has to be chosen sufficiently large such that at least the CBF condition in~\eqref{eq:cbf_lie_derivative} is feasible for all $x \in \set{C}$. 

Furthermore, while the safety filter is formulated for a continuous-time system, its implementation %
in real-world scenarios is in discrete time.
This yields a sampled-data setting since we can only apply control inputs at time intervals with duration $\Delta t > 0$. 
The state $x(t_k + \tau)$ is computed under the application of $u_{\text{cert},k}$ at time $t_k$ on the interval $\tau \in \left[0, \Delta t\right)$. 
At time $t_k$, the control input $u_{\text{cert},k} \in \set{U}_{h,\gamma}(x(t_k))$ is certified for state $x(t_k)$ using~\eqref{eqn:cbf_qp}.

In the following, we discuss the consequences of the above. 
The sampled-data implementation affects the safety filter at all states $x$ for all policies $\pi$.  
This in itself can already lead to unsafe states, $x \notin \set{C}$. While $u_{\text{cert},k}$ is certified to be a safe control input at time $t_k$, it may not be a safe control input for $x(t_k + \tau)$ for $\tau \in \left(0, \Delta t \right)$.
Now consider a state $\bar{x}$ where the safety filter is inactive and $\bar{x} \in \partial \set{C}$. 
In this case, any nonzero control input $u$ contributes to $\dot{x}$ at $\bar{x}$ pointing outside of the safe set $\set{C}$. In combination with a sampled-data implementation, this can lead to taking a step in an unsafe direction for a time step of $\Delta t > 0$. 
This can be seen from the following example. Consider a double integrator system of the form $\dot{x} =  Ax + Bu$. The CBF candidate is $h(x) = 1 - x^\intercal P x$, where $P = \text{diag}(p_1, p_2) \succ 0$. The resulting safe set is $\set{C} = \{x \in \R^2 \:\vert\: h(x) \geq 0\}$, which is an ellipsoid. Exact-discretization with a zero-order hold control input $u(t_k) = u_k \in \R$ yields
        $x(t_k + \Delta t) = x_{k+ 1} = \hat{A} x_k + \hat{B} u_k $,
where $\hat{A} = \exp{(A \Delta t)}$, $\hat{B} = \left(\int_{0}^{\Delta t} \exp{(A \tau)}d \tau \right) B$, and $\Delta t > 0$. 
For $x_k = \left[ -1/\sqrt{p_1} \quad 0 \right]^\intercal$,
$h(x_k) = 0$ and $\mathcal{L}_gh(x_k) = 0$, the change in $h$ is $h(x_{k + 1}) - h(x_k) = h(x_{k + 1}) = \eta u_k^2 + \zeta u_k$ with $\eta = - (\Delta t)^2 (p_2 + (\Delta t)^2 p_1 /4)$ and $\zeta = (\Delta t)^2 \sqrt{p_1}$,
which is quadratic in $u_k$. 
For any $u_{k}< 0$ or $u_{k} > \frac{ 4 \sqrt{p_1}}{4 p_2 + p_1 (\Delta t)^2}$, under the assumption of $u_k \in \set{U}$, this yields $h(x_{k + 1}) < 0$, such that the system leaves the safe set, see~\autoref{fig:vector-epsilon}. 
Finally, we examine the scenario where, in addition to the above, the desired control policy also minimizes $h(x_{k + 1})$.  
For the double integrator example, either the largest or the smallest feasible control input $u_{\text{cert}, k} \in \set{U}$ 
minimizes the value of $h$ at time step $(k + 1)$. 
This highlights how inactive safety filters can lead to safe set violations. 
Once the system reaches a state at which $\lVert \mathcal{L}_gh(x) \rVert$ is sufficiently large, the safety filter becomes active again and guides the system back into the safe set $\set{C}$. If this behavior repeats multiple times, the system is exhibiting chattering, which induces high-frequency control inputs that are typically undesirable in real-world systems. 
We briefly note that these effects also apply to other states where the safety filter is inactive and not just $\mathcal{L}_g h(x) = 0$~(e.g., for states where $\lVert \mathcal{L}_g h(x) \rVert \leq \epsilon(x)$). Furthermore, states with active safety filters that are in a neighborhood of the boundary may also yield safety violations.  

\subsubsection{Prevalence of Inactive CBF Safety Filters} 
After discussing the causes and the undesirable effects of inactive safety filters, we next highlight for which systems and CBF candidates the problem of inactive safety filters may be encountered. As it turns out, inactive safety filters are common for many physical systems with compact safe sets $\set{C}$. %

Our next example shows that a commonly chosen CBF in the literature does not have relative degree $\rho=1$ but rather least relative degree $\rho=1$ on $\set{D}$, which results in safety filter inactivity. 
    Consider linear system dynamics of the form $\dot{x} =  Ax + Bu$. The CBF candidate is $h(x) = 1 - x^\intercal P x$, where $P$ is positive definite. For the first-order Lie derivative with respect to $g$, we have
        $\mathcal{L}_g h(x) = 
        - 2 x^\intercal P B\,,$
    which is state-dependent.
    According to~\autoref{def:rel-degree}, the set of states where the relative degree is one and not equal to one are $\set{D}_{\rho = 1} = \{ x \in \R^2 \,\vert\, x^\intercal P B \neq 0\}$ and $\set{D}_{\rho \neq 1} = \{ x \in \R^2 \,\vert\, x^\intercal P B = 0\}$, respectively.
    In this example, $\set{D}_{\rho \neq 1}$ 
    is the set of states that are orthogonal to $PB$.

Although in this simple case, the Lie derivative along $g$ is linear in $x$, the Lie derivative along $g$ may be an arbitrary nonlinear function of $x$ in general.
This example highlights that the inactivity has to be carefully analyzed for nonlinear systems.
The inactivity 
can be evaluated by determining the appropriate Lie derivatives. However, it is generally not possible to analytically determine $\set{D}_{\rho = 1}$,
except for certain simple cases. Instead, sampling techniques must be used, or a set of nonlinear programs must be solved.

Building on the results in~\cite{Alyaseen2025TAC}, we can show that there exist states such that the associated safety filter will be inactive. 
\begin{corollary}
    \label{theorem:cbf_over_compact_sets}
Consider dynamics in~\eqref{eq:nonlinear_affine_control} with constant $g(x) = B \in \R^{n \times m}$ and a compact set $\set{C}\subseteq \set{D}$ that is parameterized as the zero-superlevel set of a continuously differentiable function $h:\set{D}\mapsto\set{R}$ with $\partial \set{C} = \{x\in\set{D}\:\:|\:\: h(x) = 0\}$.
The gradient satisfies $\partial h(x) / \partial x \neq 0 $ for all $x\in\partial \set{C}$. 
For any $d$-level set $\partial \set{C}_d$ with $d \in (0, h_{\max})$, there exists at least one point $x \in \partial \set{C}_d$, where $\partial h(x) / \partial x \perp B$ such that the associated safety filter in~\eqref{eqn:cbf_qp} is inactive. 
\end{corollary}
\begin{proof}
    The proof follows by repeating the argument from~\cite{Alyaseen2025TAC} for each level set. 
\end{proof}

In our example above, each level set had exactly two states such that $\mathcal{L}_g h(x) = 0$. 
In practice, many systems satisfy the condition of a constant $g(x)$, and the requirement of a compact safe set is ubiquitous in the field of safety-critical control. 
While the result only applies to systems with a constant input matrix, safety filter inactivity can also exist for systems with a state-dependent input matrix $g(x)$. 

\subsection{Multiple CBFs to Prevent Inactive Safety Filters}

In this section, we propose a method to circumvent inactive safety filters. 
In particular, we leverage multiple CBFs to achieve a safety filter design that is active for all states $x$ in the safe set. 
Based on the previous section, one requirement for an active safety filter is that $\mathcal{L}_g h(x) \neq 0$ for all $x$ or, in other words, $\xinactive = \emptyset$. We can further require that $\lVert \mathcal{L}_g h(x) \rVert \geq \epsilon$, where $\epsilon > 0$ can be chosen such that at least one CBF condition is not redundant. As this may not be satisfied by a single CBF~(see \autoref{theorem:cbf_over_compact_sets}), 
we use a set of $K\in \N$ CBF candidates instead. 
The safe set is then the positive superlevel set of all CBFs: $\set{C}_{\set{Z}_{1,K}}$.
We emphasize that the safe set resulting from intersecting multiple CBFs can be designed such that any desired safe set $\set{X}$
may be approximated.  
Using this notation, we also define $\set{D}_{\epsilon} = \set{D} \setminus \bigcap_{i \in \set{Z}_{1,K}} \set{D}_\epsilon^i$ with $ \xineps^i = \{x \in \R^n \, \vert \, \lVert \mathcal{L}_g h_i(x) \rVert \leq \epsilon \}$ and $\set{U}_{(h_i, \gamma_i)_{i \in \set{Z}_{1,K}}}(x) = \bigcap_{i \in \set{Z}_{1,K} } \set{U}_{(h_i, \gamma_i)}(x)$ with $\set{U}_{(h_i, \gamma_i)}(x)  = \{ u \in \set{U} \, \vert \, \alpha_i(x) u \leq \beta_i(x) \}$. 

We formally state how multiple CBFs can prevent the safety filter inactivity issue in the following theorem: 
\begin{theorem}[$\epsilon$-active CBF safety filter]
\label{th:constrained-input}
    Consider the system dynamics in \eqref{eq:nonlinear_affine_control}, a compact control input set $\set{U}$, and an associated set of multiple CBFs 
    $h_i: \set{D} \to \R $, with $i \in \set{Z}_{1, K}$ and $K \in \N$ and a given small positive number $\epsilon > 0$. 
    Let the systems consisting of the system dynamics and each CBF has a least relative degree $\rho_i = 1$. 
    If the intersection of all states for which the norm of the Lie derivative along $g$ is less than or equal to $\epsilon$ for each CBF is empty, 
        $\xineps = \cap_{i \in \set{Z}_{1, K}} \xineps^i   
        =  \emptyset \,,$
then, for all $x \in \set{D}$, it holds that $\max_{i \in \set{Z}_{1,K}} \lVert \mathcal{L}_g h_i(x) \rVert > \epsilon$. Furthermore, if 
\begin{equation*}
    \max_{i \in \set{Z}_{1,K}} \lVert \mathcal{L}_g h_i(x) \rVert > \bar{\epsilon} := \max_{i\in \set{Z}_{1,K}} \frac{M_{f_i} + \gamma_i(h_{i, \max})}{ r- \lVert u_c \rVert}\,,
\end{equation*}
where $M_{f_i} > 0$ with $M_{f_i} \geq \lVert \mathcal{L}_f h_i(x) \rVert$, and $u_c$ and $r > 0$ are the Chebyshev center and radius, respectively, such that $\set{B}(u_c, r)$ is the largest ball satisfying $\set{B}(u_c, r) \subseteq \set{U}$, and $r - \lVert u_c \rVert > 0$, 
then 
        $\set{U}_{(h_i, \gamma_i)_{i \in \set{Z}_{1,K}}}(x)
        \subset \set{U}$.
\end{theorem}
\begin{proof}
    The result that, $\max_{i \in \set{Z}_{1,K}} \lVert \mathcal{L}_g h_i(x) \rVert > \epsilon$, follows directly from the definition of the intersection. For each $x \in \set{D}$, we know that, for at least one $i$, $\mathcal{L}_g h_i(x)$ is nonzero, $\exists i \in \set{Z}_{1, K}, \lVert \mathcal{L}_g h_i(x) \rVert > \epsilon$. Therefore, there exists at least one nonzero half-space constraint such that $\bigcap_{i \in \set{Z}_{1, K}} \{ u \in \R^m \, \vert \, \alpha_i(x) u \leq \beta_i(x) \} \subset \R^m$.  
    As stated at the beginning of the section, to guarantee that the safety filter is always active, $\set{U}_{(h_i, \gamma_i)_{i \in \set{Z}_{1,K}}}(x)
        \subset \set{U}$, we require for at least one $i \in \set{Z}_{1,K}$ that
        $\lVert \mathcal{L}_g h_i(x) \rVert > \epsilon_i(x)$.
    We upper-bound the right-hand side
    using $\sigma_{\set{U}}(\alpha_i) \geq \alpha_i^\intercal u_c + r \lVert \alpha_i \rVert$
    as follows:
    \begin{align*}
        \epsilon_i(x) \leq 
        \frac{M_{f_i} + \gamma_i(h_{i, \max})}{\alpha^\intercal_i(x) u_c + r } 
        \leq \frac{(M_{f_i} + \gamma_i(h_{i, \max})) \lVert \mathcal{L}_g h_i(x) \rVert}{ ( r- \lVert u_c \rVert) \lVert \mathcal{L}_g h_i(x) \rVert } \,.  
    \end{align*}
    Since $r - \lVert u_c \rVert  > 0$
    and $\lVert \mathcal{L}_g h_i(x) \rVert \neq 0$, 
    maximizing over all $i \in \set{Z}_{1, K}$ yields the desired result.
\end{proof}

In the above, we have shown how we can use multiple CBFs to prevent safety filter inactivity. In particular, the parameter $\epsilon > 0$ trades off conservatism in the feasible inputs and safety filter inactivity. If $\epsilon$ is sufficiently large, the safety filter is active for all states $x \in \set{C}$. 
The derived value for $\epsilon$ can lead to conservative or infeasible input sets in practice. The result provides a theoretical upper bound for $\epsilon$ that may guide the selection of the $K$ CBFs. 
Finally, if $\set{D}_\epsilon = \emptyset$ for some $\epsilon > 0$, then the associated safety filter~(even with $\set{U} = \R^m$) prevents the issue of unbounded $u_{\text{cert},k}$, which was analyzed in~\cite{Alyaseen2025TAC}.  

\subsection{Synthesizing Multiple CBFs}
\label{sec:synthesis}
In this section, we extend the CBF synthesis problem to multiple CBFs that account for the system's least relative degree.  
The goal is to synthesize the largest intersection of $K$ zero-superlevel sets $\set{C}_{\{i\}} = \{ x \in \R^n \vert h(x;\theta_i) \geq 0 \}$ that are contained inside a desired set of state constraints $\set{X}$. These multiple CBFs are required to satisfy the CBF condition, and the Lie derivative $\mathcal{L}_\vec{g} h_i(x)$ may never be close to zero. We express this problem using the following optimization problem:  
    \begin{align}
    \raisetag{70pt} 
    \label{eqn:multi-cbf-synthesis-general}
	\begin{split}
	\underset{\{\theta_i \in \R^{p}, \phi_i \in \R^{q}\}_{i=1}^{K}}{\text{max}} & \quad \text{Vol}(\set{C}_{\set{Z}_{1, K}}) \\ \text{s.t.} 
    & \quad \set{C}_{\set{Z}_{1, K}} = \bigcap_{i \in \set{Z}_{1, K}} \set{C}_{\{i\}}\subseteq \set{X} 
    \,, \\
    & \quad \max_{i \in \set{Z}_{1, K}}\lVert \mathcal{L}_\vec{g} h(x; \theta_i) \rVert \geq \epsilon\,, \forall x \in \set{C}_{\set{Z}_{1, K}} 
    \,, \\
    & \quad \bigcap_{i \in \set{Z}_{1, K}}\set{U}_{(h_i, \gamma_i)}(x) \neq \emptyset\,, \forall x \in \set{C}_{\set{Z}_{1, K}}\,,
	\end{split}
    \end{align}
where $\theta_i$ and $\phi_i$ parameterize $h_i$ and $\gamma_i$, respectively, $\text{Vol}$ is the set's volume, $\epsilon > 0$, and $K > 1$ are user-defined parameters. 

The above optimization problem 
includes multiple semi-infinite constraints, and there is no analytical expression for the volume of the intersection of general superlevel sets. 
Therefore, we propose to use a sampling-based approach to determine a feasible solution. Our proposed sampling approach may also be implemented as an evolutionary algorithm, which was proposed for synthesizing a single CBF in~\cite{Wei2023}. 

To reduce the search space, we sample the function parameters from compact sets $\Theta$ and $\Phi$, respectively. For every sampling iteration $k$, we sample parameters $\theta_{i,k} \in \Theta$ and check if $\set{C}_{\set{Z}_{1, K},k} \subseteq \set{X}$. If $\set{X}$ itself is the zero superlevel set of a continuously differentiable function $h_{\set{X}}(x)$, this can be verified by 
        $\min_{x \in \set{C}_{\set{Z}_{1, K},k}}~ h_{\set{X}}(x) \geq 0 \,.$
Then for each function $h_{i,k}$, we determine the norm on the Lie derivative and test if $\max_{i \in \set{Z}_{1,K}} \lVert \mathcal{L}_g h(x; \theta_{i,k}) \rVert \geq \epsilon$ for a sufficiently large set of sampled states $x \in \set{C}_{\set{Z}_{1,K},k}$. Alternatively, we can check the sufficient condition $\min_{x \in \set{C}_{\set{Z}_{1,K},k}} \lVert \mathcal{L}_g h(x; \theta_{i,k}) \rVert \geq \epsilon$ for all $i \in \set{Z}_{1,K}$ using a nonlinear programming solver. Finally, we sample a parameter $\phi_{i, k_j} \in \Phi$ and need to verify the feasibility of the set of constraints $\mathcal{L}_\vec{f} h(\x, \theta_{i,k}) + \mathcal{L}_\vec{g} h(\x; \theta_{i,k}) \u \geq - \gamma(h(\x;\theta_{i,k}); \phi_{i,k_j})$, $ \forall x \in \set{C}_{\set{Z}_{1, K},k}$, $ \forall u \in \set{U}$, $ \forall i \in \set{Z}_{1,K} \,.$ This is a challenging problem for general nonlinear systems and CBFs~\cite{brunke-lcss-2024}, which often can only be determined by sampling sufficiently many points in $\set{C}_{\set{Z}_{1, K},k}$ or the dynamics and constraints have a special structure that simplifies the computation~\cite{Dai2022}. In case the feasibility cannot be guaranteed, we can sample a new parameter $\phi_{i, k_{j+1}} \in \Phi$ and execute the feasibility check again. For every successful sample, we approximate the volume of set $\set{C}_{\set{Z}_{1, K}}$, for example, through its bounding box or rejection sampling. Finally, we pick the parameters $\theta_{i,k^*}$ and $\phi_{i,k^*}$, that lead to the largest surrogate of the true volume of $\set{C}_{\set{Z}_{1, K},k^*}$. 

\subsection{Multiple CBFs for Safety under Sampled-Data Control}
In this section, we analyze when a safety filter with multiple CBFs yields safety for discrete-time implementations~(i.e., using sampled-data control). 

We define $\bar{u} = \max_{u \in \set{U}} ~ \lVert u \rVert$ as the maximum feasible control input. 
A Lipschitz constant for the control affine system's dynamics in~\eqref{eq:nonlinear_affine_control} can be determined as follows:
    \begin{align*}
        \lVert f(x) + g(x) u - f(y) - g(y) u \rVert
        \leq (L_f + L_g \bar{u}) \lVert x - y \rVert \,.
    \end{align*}
We define $L := (L_f + L_g \bar{u})$.
Using this Lipschitz constant, we can bound the difference between the current and future states after applying control input $u$ for sampling time $\Delta t$. Consider the following slightly adapted result to explicitly account for the state and control input dependency: 
\begin{proposition}[Bounded sampled-and-hold deviation~\cite{yang2019}]
\label{thm:bounded-error}
    Consider the control-affine system defined in \eqref{eq:nonlinear_affine_control}. Suppose at time $t_0 \geq 0$, a constant control input $u$ is applied to the system for a period of $\Delta t = t - t_0 \geq 0$. The distance between the future state $x(t)$ and the state $x_{t_0} = x(t_0)$ is $e(t, x_{t_0}, u) = \lVert x(t) - x_{t_0} \rVert$. Let 
        $\Bar{e}(t, x_{t_0}, u) := \frac{1}{L} \lVert f(x_{t_0}) + g(x_{t_0})u\rVert \left( \exp (L \Delta t) - 1 \right)$. %
    Then, 
        $e(t, x_{t_0}, u) \leq \Bar{e}(t, x_{t_0}, u)$
    for all $ t \geq t_0$.
\end{proposition}
\begin{proof}
    The proof follows from~\cite{yang2019}.  
\end{proof}

With~\autoref{thm:bounded-error} and compact $\set{U}$,
we can show that trajectories starting in $\set{C}_{\set{Z}_{1,k}}$ will stay bounded: 
\begin{lemma}[Bounded trajectories]
\label{thm:bounded-error-max}
    The error $e(t, x, u)$ is upper-bounded for any $x \in \set{C}_{\set{Z}_{1,k}}$ and any $u \in \set{U}$ with
    \begin{equation}
            \bar{e}(t, x, u) \leq  \frac{\bar{f} + \bar{g} \bar{u} }{L}  \left( \exp (L \Delta t) - 1 \right) \label{eq:error-bound-u-max}=: \hat{e}(\Delta t) \,,
    \end{equation}
    where $\bar{f} = \max_{x \in \set{C}_{\set{Z}_{1,k}}} \lVert f(x) \rVert$, $\bar{g} = \max_{x \in \set{C}_{\set{Z}_{1,k}}} \lVert g(x) \rVert$.
\end{lemma}
\begin{proof} 
    The dynamics are upper bounded by 
        $\lVert f(x)  +  g(x)  u \rVert 
        \leq \bar{f} + \bar{g} \bar{u} \,.$
    The maxima $\bar{f}$ and $\bar{g}$ exist because $f$ and $g$ are Lipschitz continuous on the bounded set $\set{C}_{\set{Z}_{1,k}}$. Finally, the maxima upper-bound the deviation for all states $x \in \set{C}_{\set{Z}_{1,k}}$. 
    This gives the desired result.     
\end{proof}

We leverage the result in~\autoref{thm:bounded-error-max} to determine a positive sampling time $\Delta t$ that yields control invariance for the sampled-data control system. This is formalized in the following theorem: 
\begin{theorem}[Safe sampled-data control]
    \label{thm:safe-sampled-data}
    Consider the sample-and-hold system in~\eqref{eq:sampled-data-sys}. 
    Let $\set{U}$ be a compact set and let the sample-and-hold error for each CBF be bounded by $e_{i, \max} = - \frac{\gamma_i(-d_i)}{L_\pi M_i}$. If the sampling time $\Delta t > 0$ and 
    \begin{equation}
    \label{eq:delta-t-bound}
        \Delta t \leq \frac{1}{L} \ln \left(1 + \frac{\min_{i \in \set{Z}_{1, K}} \left(-\gamma_i (-d_i) \right) L}{L_\pi M_i (\bar{f} + \bar{g} \bar{u})}\right)
    \end{equation} 
    for $d_i \in \left( 0, h_{i, \max} \right)$ such that $\set{C}_d = \bigcap_{i \in \set{Z}_{1, K}} \set{C}_{d_i} = \left\{ x \in \R^n \, \vert \, h_i(x) - d_i \geq 0 \right\} \neq \emptyset$ for each $i \in \set{Z}_{1, K}$, 
    then the sample-and-hold system~\eqref{eq:sampled-data-sys} can be rendered positively control invariant on $\set{C}_{\set{Z}_{1, K}}$ using~\autoref{thm:inv-sampled-data} if the resulting control input set $\set{U}_{(h_i, \gamma_i)_{i \in \set{Z}_{1,K}}}(x) \neq \emptyset$, $ \forall x \in \set{C}_{\set{Z}_{1, K}}$. 
\end{theorem}
\begin{proof}
    As the upper bound on the sampled-data-error $\hat{e}(\Delta t_i)$ is bounded by $e_{i, \max}$ we get
    \begin{equation*}
        \frac{\bar{f} + \bar{g} \bar{u} }{L}  \left( \exp (L \Delta t_i) - 1 \right) \leq - \frac{\gamma_i(-d_i)}{L_\pi M_i} \,.
    \end{equation*}
    Due to~\autoref{thm:bounded-error} and~\autoref{thm:bounded-error-max}, satisfying this inequality yields a valid bound on all feasible $e(\Delta t_i, x, u)$ with $x \in \set{C}_{\set{Z}_{1, K}}$ and $u \in \set{U}_{(h_i, \gamma_i)_{i \in \set{Z}_{1,K}}} (x)$.  
    Then, solving for $\Delta t_i$, results in 
    \begin{equation*}
        \Delta t_i \leq \frac{1}{L} \ln \left(1 + \frac{-\gamma_i (-d_i) L}{L_\pi M_i (\bar{f} + \bar{g} \bar{u})}\right) \,.
    \end{equation*}
    Taking the minimum over all $i \in \set{Z}_{1,K}$ and moving the minimum into the logarithm gives the desired result. 
    Finally, the last condition on the resulting control input set guarantees that the resulting CBF condition yields a feasible control input at every state inside the safe set. 
\end{proof}

While this lets us compute the required sampling time $\Delta t$, in practice, it is usually fixed by the system or hardware. Thus, given $\Delta t$, we want to check if it ensures safety by determining the necessary tightening $d_i$ for each CBF, as stated in the following corollary:
\begin{corollary}
\label{thm:safe-sampled-data-cor}
    Let $\set{U}$ be a compact set and let the sample-and-hold error for each CBF be bounded by $e_{i, \max} = - \frac{\gamma_i(-d_i)}{L_\pi M_i}$. If  $d_i \in \left( 0, h_{i, \max} \right)$ for each $i \in \set{Z}_{1, K}$ satisfies 
    \begin{align*}
    d_i \geq 
    -\gamma_i^{-1} \left( \frac{L_\pi M_i  
    (\bar{f} + \bar{g} \bar{u}) }{L } \left( 1 - \exp (L) \Delta t) \right) \right)\,,
\end{align*}
    with $\Delta t > 0$ such that $\set{C}_d = \bigcap_{i \in \set{Z}_{1, K}} \set{C}_{d_i} = \left\{ x \in \R^n \, \vert \, h_i(x) - d_i \geq 0 \right\} \neq \emptyset$. 
    Then, the sample-and-hold system~\eqref{eq:sampled-data-sys} can be rendered positively control invariant on $\set{C}_{\set{Z}_{1,k}}$ if the resulting control input set $\set{U}_{(h_i, \gamma_i)_{i \in \set{Z}_{1,K}}}(x) \neq \emptyset$, $ \forall x \in \set{C}_{\set{Z}_{1,K}}$. 
\end{corollary}

\section{Evaluation}
\label{sec:eval}

\subsection{Simulation Example}
\label{sec:sim}
We verify our theoretical results in
simulation.
Details of the simulation setup, along with implementations, can be found at: 
\href{https://github.com/lukasbrunke/multi-cbf}{\texttt{https://github.com/lukasbrunke/multi-cbf}}.

We consider a continuous-time double integrator dynamics model $\dot{x} = A x + Bu$
with $x \in \R^2$ and $u \in \R$.
We investigate two cases: \textit{(i)} using a single CBF~(least relative degree is one) and \textit{(ii)} using two CBFs~(joint relative degree is one).
In the first case, we use the CBF $h(x) = 1 - (x - c)^\intercal P (x - c)$, where $c = \begin{bmatrix}
    0 & 0
\end{bmatrix}^\intercal$ and $P = \text{diag}(1, 2)$. 
For the second case, we solve the optimization in~\eqref{eqn:multi-cbf-synthesis-general} over $\theta_i = \begin{bmatrix}
    \operatorname{diag}(P_i)^\intercal & c_i^\intercal
\end{bmatrix}^\intercal$ and a fixed slope of a linear class-$\K_e$ function expressed as $\phi_i = 2$%
such that $\epsilon = 0.01$ and $\set{C}_{\{1,2 \}} \subseteq \set{C} = \set{X}$.
This yields a safety filter with the CBFs $h_{i}(x) = 1 - (x - c_{i})^\intercal P_{i} (x - c_{i})$, with $i \in \left\{1, 2\right\}$, where $c_{1} = \begin{bmatrix}
    0.01 & - 1.27
\end{bmatrix}^\intercal$, 
$c_{2} = \begin{bmatrix}
    0.02 & 1.33
\end{bmatrix}^\intercal$, 
$P_{1} = \text{diag}(0.60, 0.26)$, and $P_2 = \text{diag}(0.56, 0.26)$. 
For the single CBF, we also use $\gamma(r) = \gamma_{i}(r) = 2r$. 
The unsafe control input policy is $\pi(x) = -1.0$ and the input constraint set is $\set{U} = \{ u \in \R \,\,\vert\,\, \lvert u \rvert \leq 2.0 \}$.
For both cases, we use a control frequency of $\SI{10}{\kilo \hertz}$ in the simulation. 
We use~\autoref{thm:safe-sampled-data} to determine the required tightenings $d > 0$ and $d_i > 0$, $\forall i \in \{1, 2\}$ for each CBF to apply our results for sampled data systems. 
Our calculations yield $d_1 = 0.0002$ and $d_2 = 0.0007$. The high control frequency leads to a small tightening of the safe sets. The chattering for the single-CBF approach yields a large Lipschitz constant for the closed-loop policy, such that no feasible tightening exists. Instead, we use a nominal implementation for the single-CBF case.  

    In~\autoref{fig:combined-cbf-sim-input}, we see that the single-CBF safety filter results in chattering and safe set violations~(left) and that our proposed multi-CBF safety filter successfully prevents a high-frequency control input signal~(right).
    The multiple CBFs are tightened according to our theoretical results in~\autoref{thm:safe-sampled-data} to achieve a safe closed-loop sampled-data system. 
    We highlight the states where the relative degree $\rho \neq 1$ in green~(top). Note that these states are shown using a dotted-dashed line for the single CBF case and a dashed line for each CBF in the case with multiple CBFs. However, since the intersection of these sets is empty in the multiple CBF case, no states satisfy $\rho \neq 1$ for both CBFs simultaneously. This prevents the safety filter with the multiple CBFs from being inactive. 
    For both safety filters, we 
    apply an uncertified control input policy~(blue line in the bottom plots). 
    While our proposed approach never violates the safety constraints, the single-CBF approach violates the constraints~(red line in the top plots) and yields a chattering behavior as visible in the input trajectory~(bottom plot). 
    This is because the single CBF safety filter becomes inactive close to states $x \in \set{X}_{\rho \neq 1}$. For the safety filter with multiple CBFs, there is no chattering. We also highlight the input constraints~(black dashed lines) and show that the lower bound on the constraint converges to~$0$. Therefore, no chattering is possible with our proposed safety filter. 

\begin{figure}[t]
    \centering
    \include{tikz/sim-state-10khz}
\vspace{-1.4cm}
    \include{tikz/sim-input-10khz}
    \vspace{-1.0cm}
    \caption{
    Simulation comparing single-CBF safety filter with chattering (left) versus our multi-CBF safety filter preventing high-frequency control~(right). 
    The states with relative degree $\rho \neq 1$ are shown in green. Both filters use uncertified policy $\pi(x) = -1.0$. The single-CBF case exhibits constraint set violations and input chattering due to loose constraints near $x \in \set{D}_{\rho \neq 1}$~(not shown to avoid clutter). Our multi-CBF approach ensures safety filter activity and prevents chattering.
    }
    \label{fig:combined-cbf-sim-input}
\end{figure}

Solving QPs at high rates (e.g., $\SI{10}{\kilo \hertz}$) is often impractical on resource-constrained hardware. We therefore repeated the simulation at $\SI{100}{\hertz}$, for which theoretical guarantees no longer hold, as tightened safe sets lead to infeasible QPs. However, similar behavior was observed without tightening, and results~(not shown) closely match those in~\autoref{fig:combined-cbf-sim-input}. While theoretical conditions are not met, using multiple CBFs still improves closed-loop performance over a single CBF. This suggests that lower control frequencies may suffice in practice.

\subsection{Quadrotor Experiments}
\label{sec:quad}
We verify our proposed method of multiple CBFs to prevent inactive safety filters on a real-world nonlinear quadrotor system. A picture of the experiment is shown in~\autoref{fig:money-figure}, and a video can be found at this link: \href{http://tiny.cc/multi-cbf}{\texttt{http://tiny.cc/multi-cbf}}.

We leverage the quadrotor's roll, pitch, yaw, and collective thrust interface,
where the thrust is commanded through pulse-width modulation~(PWM).
As the $y$-position is controlled to be 0, the reduced system state is $\vec{x} = \begin{bmatrix}
    p_x & p_z & \theta & v_x & v_z
\end{bmatrix}^\intercal \in \R^5$, where $\theta$ is the quadrotor's pitch angle, and the commanded control input is $u = \begin{bmatrix}
    \theta_{\text{desired}} & F_{\text{desired}}
\end{bmatrix}^\intercal \in \R^2$. 
The identified nonlinear model is
$\dot{x} = f(x) + g(x) u$, 
where $f(x) = \begin{bmatrix}
        v_x & v_z & \alpha_1 \theta & \beta_1 \sin(\theta) & \beta_1 \cos(\theta) -g
    \end{bmatrix}^\intercal$, $g(x) = \text{diag}(g_1, g_2(x))$, 
    $g_1 = \begin{bmatrix}
    0 & 0 & \alpha_2
\end{bmatrix}^\intercal$, $g_2(x) = \beta_2 \begin{bmatrix}
    \sin(\theta) & \cos(\theta)
\end{bmatrix}^\intercal$, $\alpha_1 = -\alpha_2 = -60.00$, $\beta_1 = 4.60$, and $\beta_2 = 15.40$. 
We use a position controller $\pi(t, x)$ 
to track an unsafe reference trajectory. 

In the first case, we use a single CBF $h(x) = 1 - (x - c)^\intercal P (x - c)$, where $c = \begin{bmatrix}
    0.0 & 1.52 & 0.0 & 0.0 & 0.0
\end{bmatrix}^\intercal$ and $P = \text{diag}(0.0, 0.7, 0.0, 0.0, 2.0)$. In the second case, we apply a safety filter with the CBFs $h_{i}(x) = 1 - (x - c_{i})^\intercal P_{i} (x - c_{i})$, with $i \in \left\{1, 2\right\}$, where $c_{1} = \begin{bmatrix}
    0.0 & 1.55 & 0.0 & 0.0 & 2.5
\end{bmatrix}^\intercal$, $c_{2} = \begin{bmatrix}
    0.0 & 1.55 & 0.0 & 0.0 & -2.5
\end{bmatrix}^\intercal$, and $P_{1} = P_2= \text{diag}(0.0, 0.25, 0.0, 0.0, 0.1)$. The CBFs are chosen such that $\set{C}_{\{ 1, 2\}} \subseteq \set{C}= \set{X}$.
For all CBFs, we use $\gamma_{}(r) = \gamma_{i}(r) = 3r$ and we implement the safety filters at a frequency of $\SI{60}{\hertz}$. 

For both safety filters, we initialize the system at approximately $x = \begin{bmatrix}
   0.0 & 1.0 & 0.0 & 0.0 & 0.0
\end{bmatrix}^\intercal$ and apply the same uncertified policy $\pi(t, x)$~(blue line in~\autoref{fig:multi-cbf-real-input}).
We show that the safety filter with multiple CBFs~(plots on the right-hand side in~\autoref{fig:multi-cbf-real-input}) achieves safety. In contrast, the safety filter with a single CBF~(plots on the left-hand side) violates the safe set. In addition, it yields an oscillatory behavior in the certified control input and the closed-loop state trajectory. 
We highlight the states where the $\mathcal{L}_g h(x) = 0$ in green. These states are not visible in the state space plot on the right (the case with multiple CBFs). However, since the intersection of these sets is empty in the multiple CBF case, no states satisfy $\mathcal{L}_g h_1(x) = 0 \land \mathcal{L}_g h_2(x) = 0$ for both CBFs simultaneously. This prevents the safety filter with the multiple CBFs from being inactive and keeps the system inside the safe set~$\set{C}_{\{1, 2 \}}$.
Due to the choice of the class $\K_e$ function, the quadrotor stays further away from the boundary of the safe set compared to the system in the simulation result, see~\autoref{fig:money-figure}. 

\begin{figure}[t]
    \centering
    \include{tikz/real-state}
\vspace{-1.4cm}

    \include{tikz/real-input}
    \vspace{-1.0cm}
    \caption{
    Real-world quadrotor comparison: single-CBF safety filter violating the safe set~(left) versus our multi-CBF safety filter preventing violations~(right). Both use $\SI{60}{\hertz}$ control frequency with identical uncertified tracking policy $\pi(t, x)$. The multi-CBF filter achieves safety, while the single-CBF case exhibits oscillatory behavior and safe set violations. The states with relative degree $\rho \neq 1$ are shown in green. 
    }
    \label{fig:multi-cbf-real-input}
\end{figure}

\section{Conclusion}
\label{sec:conclusion}
In this work, we investigated the issue of invalid relative degree assumptions, how they result in inactive CBF safety filters, and their negative impact on discrete-time implementations such as chattering and safe set violations.
We characterize when a CBF safety filter is inactive, analyze the direction and magnitudes of control inputs at states of inactivity, and discuss the implications of the discrete-time implementation. 
We present a method that leverages multiple CBFs to prevent inactive safety filters and
an upper bound on the sampling time to ensure safety in discrete-time implementations.
From the simulation and quadrotor experiments, we verify and show that our proposed method efficiently mitigates chattering and safe set violations caused by inactive safety filters.





\ifreview 
\else
    \section*{References}
\fi

\bibliographystyle{./IEEEtranBST/IEEEtran}
\bibliography{./IEEEtranBST/IEEEabrv,./root}

\begin{thebibliography}{10}
\providecommand{\url}[1]{#1}
\csname url@rmstyle\endcsname
\providecommand{\newblock}{\relax}
\providecommand{\bibinfo}[2]{#2}
\providecommand\BIBentrySTDinterwordspacing{\spaceskip=0pt\relax}
\providecommand\BIBentryALTinterwordstretchfactor{4}
\providecommand\BIBentryALTinterwordspacing{\spaceskip=\fontdimen2\font plus
\BIBentryALTinterwordstretchfactor\fontdimen3\font minus \fontdimen4\font\relax}
\providecommand\BIBforeignlanguage[2]{{%
\expandafter\ifx\csname l@#1\endcsname\relax
\typeout{** WARNING: IEEEtran.bst: No hyphenation pattern has been}%
\typeout{** loaded for the language `#1'. Using the pattern for}%
\typeout{** the default language instead.}%
\else
\language=\csname l@#1\endcsname
\fi
#2}}

\bibitem{DSL2021}
L.~Brunke, M.~Greeff, A.~W. Hall, Z.~Yuan, S.~Zhou, J.~Panerati, and A.~P. Schoellig, ``Safe learning in robotics: From learning-based control to safe reinforcement learning,'' \emph{Annual Review of Control, Robotics, and Autonomous Systems}, vol.~5, pp. 411--444, 2022.

\bibitem{jaime-survey-2024}
K.-C. Hsu, H.~Hu, and J.~F. Fisac, ``The safety filter: A unified view of safety-critical control in autonomous systems,'' \emph{Annual Review of Control, Robotics, and Autonomous Systems}, vol.~7, no.~1, p. null, 2024.

\bibitem{Nagumo1942berDL}
M.~Nagumo, ``{{\"U}ber die Lage der Integralkurven gew{\"o}hnlicher Differentialgleichungen},'' \emph{Proc. of the Physico-Mathematical Society of Japan, 3rd Series}, vol.~24, pp. 551--559, 1942.

\bibitem{WIELAND2007462}
P.~Wieland and F.~Allgöwer, ``Constructive safety using control barrier functions,'' \emph{IFAC Proceedings Volumes}, vol.~40, no.~12, pp. 462--467, 2007, 7th IFAC Symposium on Nonlinear Control Systems.

\bibitem{ames2019a}
A.~D. Ames, S.~Coogan, M.~Egerstedt, G.~Notomista, K.~Sreenath, and P.~Tabuada, ``Control barrier functions: Theory and applications,'' in \emph{Proc. of the European Control Conf. (ECC)}, 2019, pp. 3420--3431.

\bibitem{brunke-lcss-2024}
L.~Brunke, S.~Zhou, M.~Che, and A.~P. Schoellig, ``Optimized control invariance conditions for uncertain input-constrained nonlinear control systems,'' \emph{IEEE Control Systems Letters}, vol.~8, pp. 157--162, 2024.

\bibitem{Glotfelter2017}
P.~Glotfelter, J.~Cortés, and M.~Egerstedt, ``Nonsmooth barrier functions with applications to multi-robot systems,'' \emph{IEEE Control Systems Letters}, vol.~1, no.~2, pp. 310--315, 2017.

\bibitem{khalil2002}
H.~K. Khalil, \emph{Nonlinear Systems}, ser. Pearson Education.\hskip 1em plus 0.5em minus 0.4em\relax Prentice Hall, 2002.

\bibitem{ames2017}
A.~D. Ames, X.~Xu, J.~W. Grizzle, and P.~Tabuada, ``Control barrier function based quadratic programs for safety critical systems,'' \emph{IEEE Transactions on Automatic Control}, vol.~62, no.~8, pp. 3861--3876, 2017.

\bibitem{nguyen-exp-cbf-2016}
Q.~Nguyen and K.~Sreenath, ``Exponential control barrier functions for enforcing high relative-degree safety-critical constraints,'' in \emph{2016 American Control Conference (ACC)}, 2016, pp. 322--328.

\bibitem{xiao2021high}
W.~Xiao and C.~Belta, ``High-order control barrier functions,'' \emph{IEEE Transactions on Automatic Control}, vol.~67, no.~7, pp. 3655--3662, 2021.

\bibitem{JANKOVIC2018359}
M.~Jankovic, ``Robust control barrier functions for constrained stabilization of nonlinear systems,'' \emph{Automatica}, vol.~96, pp. 359--367, 2018.

\bibitem{Alyaseen2025TAC}
M.~Alyaseen, N.~Atanasov, and J.~Cortés, ``Continuity and boundedness of minimum-norm cbf-safe controllers,'' \emph{IEEE Transactions on Automatic Control}, vol.~70, no.~6, pp. 4148--4154, 2025.

\bibitem{Tan2022TAC}
X.~Tan, W.~S. Cortez, and D.~V. Dimarogonas, ``High-order barrier functions: Robustness, safety, and performance-critical control,'' \emph{IEEE Transactions on Automatic Control}, vol.~67, no.~6, pp. 3021--3028, 2022.

\bibitem{lukas-acc-2024}
L.~Brunke, S.~Zhou, M.~Che, and A.~P. Schoellig, ``Practical considerations for discrete-time implementations of continuous-time control barrier function-based safety filters,'' in \emph{2024 American Control Conference (ACC)}, 2024, pp. 272--278.

\bibitem{FRIDMAN20041441}
E.~Fridman, A.~Seuret, and J.-P. Richard, ``Robust sampled-data stabilization of linear systems: an input delay approach,'' \emph{Automatica}, vol.~40, no.~8, pp. 1441--1446, 2004.

\bibitem{Agrawal-RSS-17}
A.~Agrawal and K.~Sreenath, ``Discrete control barrier functions for safety-critical control of discrete systems with application to bipedal robot navigation,'' in \emph{Proceedings of Robotics: Science and Systems}, Cambridge, Massachusetts, July 2017.

\bibitem{singletary2020-sampled-data}
A.~Singletary, Y.~Chen, and A.~D. Ames, ``Control barrier functions for sampled-data systems with input delays,'' in \emph{2020 59th IEEE Conference on Decision and Control (CDC)}, 2020, pp. 804--809.

\bibitem{breeden2022-sampled}
J.~Breeden, K.~Garg, and D.~Panagou, ``Control barrier functions in sampled-data systems,'' \emph{IEEE Control Systems Letters}, vol.~6, pp. 367--372, 2022.

\bibitem{yang2019}
G.~Yang, C.~Belta, and R.~Tron, ``Self-triggered control for safety critical systems using control barrier functions,'' in \emph{2019 American Control Conference (ACC)}, 2019, pp. 4454--4459.

\bibitem{bahati2024}
G.~Bahati, P.~Ong, and A.~D. Ames, ``Sample-and-hold safety with control barrier functions,'' in \emph{2024 American Control Conference (ACC)}, 2024.

\bibitem{fpb-cdc2023}
F.~{Pizarro Bejarano}, L.~Brunke, and A.~P. Schoellig, ``Multi-step model predictive safety filters: Reducing chattering by increasing the prediction horizon,'' in \emph{Proc. of the IEEE Conf. on Decision and Control (CDC)}, 2023.

\bibitem{parwana2023feasible}
H.~Parwana, M.~Black, B.~Hoxha, H.~Okamoto, G.~Fainekos, D.~Prokhorov, and D.~Panagou, ``Feasible space monitoring for multiple control barrier functions with application to large scale indoor navigation,'' 2023.

\bibitem{aali2022}
M.~Aali and J.~Liu, ``Multiple control barrier functions: An application to reactive obstacle avoidance for a multi-steering tractor-trailer system,'' in \emph{2022 IEEE 61st Conference on Decision and Control (CDC)}, 2022, pp. 6993--6998.

\bibitem{Ames2014}
A.~D. Ames, J.~W. Grizzle, and P.~Tabuada, ``{Control barrier function based quadratic programs with application to adaptive cruise control},'' \emph{Proc. of the IEEE Conf. on Decision and Control (CDC)}, pp. 6271--6278, 2014.

\bibitem{Kolathaya2019}
S.~Kolathaya and A.~D. Ames, ``Input-to-state safety with control barrier functions,'' \emph{IEEE Control Systems Letters}, vol.~3, no.~1, pp. 108--113, 2019.

\bibitem{Wei2023}
T.~Wei, S.~Kang, W.~Zhao, and C.~Liu, ``Persistently feasible robust safe control by safety index synthesis and convex semi-infinite programming,'' \emph{IEEE Control Systems Letters}, vol.~7, pp. 1213--1218, 2023.

\bibitem{Dai2022}
H.~Dai and F.~Permenter, ``Convex synthesis and verification of control-{L}yapunov and barrier functions with input constraints,'' in \emph{Proc. of the IEEE American Control Conf. (ACC)}, 2023.

\end{thebibliography}

\ifreview 
\else
    \vspace{-5cm}
\fi

\end{document}